\documentclass[10pt,conference,twocolumn]{IEEEtran}
\usepackage{bbm}
\usepackage{eucal}
\usepackage{graphicx}
\usepackage{times}
\usepackage{latexsym}
\usepackage{amsmath}
\usepackage{amssymb}
\usepackage{epsfig}
\usepackage{eucal}
\usepackage{cite}
\usepackage{algorithmic}
\usepackage{algorithm}
\usepackage{subfigure}

\IEEEoverridecommandlockouts

\begin{document}

\title{Network Localization on Unit Disk Graphs\thanks{Supported by National Science Foundation Under Grants No. CCF-0635199, CCF-0829888, CMMI-0928092, and OCI-1133027. This paper appears in IEEE GLOBECOM 2011, Houston, TX.}}
\author{\begin{tabular}{ccc}
 Phisan Kaewprapha & Jing Li (Tiffany) & Nattakan Puttarak\\
Dept of ECE, Lehigh University & Dept of ECE, Lehigh University & Dept of ECE, Lehigh University\\
Bethlehem, PA 18015, USA& Bethlehem, PA 18015, USA& Bethlehem, PA 18015, USA\\
phk205@lehigh.edu &  jingli@ece.lehigh.edu &  nap205@lehigh.edu
\end{tabular}\vspace{-0.3cm}
}

\maketitle


\begin{abstract}
We study the problem of cooperative localization of a large network of nodes in integer-coordinated unit disk graphs, a simplified but useful version of general random graph. Exploiting the property that the radius $r$ sets clear cut on the connectivity of two nodes, we propose an essential philosophy that ``no connectivity is also useful information just like the information being connected'' in unit disk graphs. Exercising this philosophy, we show that the conventional network localization problem can be re-formulated to significantly reduce the search space, and that global rigidity,  a necessary and sufficient condition for the existence of unique solution in general graphs, is no longer necessary. While the problem is still NP-hard, we show that a (depth-first) tree-search algorithm with memory $O(N)$ ($N$ is the network size) can be developed, and for practical setups, the search complexity and speed is very manageable, and is magnitudes less than the conventional problem, especially when the graph is sparse or when only very limited  anchor nodes are available.
\end{abstract}

\section{Introduction and Motivation}
\label{sec:intro} 
Localization using GPS is a mature technology that has been deployed and proven very useful in a myriad of applications in  military, commercial and personal uses. In other scenarios, like wireless sensor networks where information about locations is as important as measuring data itself, GPS-equipped devices may not always be practical, since the sensing devices are often limited by the cost,   power or lack of the line-of-sight from the satellites. 

This paper considers a general scenario involving a network system, where some {\it anchor} nodes already know their exact locations (through GPS or other methods), while the rest ({\it non-anchor} nodes) do not and wait to be localized. While the GPS system can be viewed as a single-hop method, network localization is usually classified as multi-hop localization. 
Distance information between nodes are often attainable by measuring wireless modalities, such as time-of-arrival (TOA), angle-of-arrival(AOA), received-signal-strength (RSS), or the combination of them \cite{Mao}\cite{bao}. 
In general, one may assume that no non-anchors have distance measurements from 3  or more different anchors (anchors not in the same line), because otherwise these non-anchors can be quickly localized through  geometric trilateration or triangulation (single-hop method) and become anchors. 

Since a non-anchor does not have direct reachability (distance information) to enough anchors to be localized by itself, they must share resources and cooperate to accomplish the task together. Many researchers formulate this multi-hop localization problem as some type of optimization problem (non-linear, convex, semi-definite-programming) where standard algorithms can be applied, usually in a centralized manner \cite{Kannan,Biswas}. Other approaches fall into the heuristic category where the distances from a non-anchor to anchors are approximated using different techniques, varying from the number of hops \cite{Bulusu} to the sum of the exact distances \cite{Niculescu}.      

In addition to algorithmic study, there is also on-going effort to lay down the theoretical building blocks of this problem. 
In the theoretical framework, the network is commonly modelled as a geometric graph, where distances are exact and represented by straight-line edges (with definite lengths) between nodes. (More system models will be discussed in the next section.) Part of the theory deals with the properties/conditions that lead to unique localization solution. For instance, (global) rigidity 
of the graph guarantees the existence of the unique realization of the graph. If one can check for the rigidity purely based on the graph connectivity (e.g. 3-connected, 6-connected) \cite{Aspnes1}, then one can immediately determine the existence of the solution without even trying to realize the graph, and there exist efficient algorithms to verify rigidity and the related conditions. Despite  such existence tests, however, realizing a rigid graph remains a challenging problem in general. 
One exception is {\it trilateration graphs} for which efficient realization algorithms exist. A trilateration graph is one such that, at any time, at least one non-anchor is connected to enough anchors (at least 3  for 2D) to determine its location and becomes an anchor, and this condition holds iteratively through out the network.
For an arbitrary geometric graph model (non-trilateration graphs), in term of the complexity to realize the graph, they are NP-Hard \cite{Breu,Aspnes2}.

In an attempt to find more effective way to tackle the problem, unit-disk-graph model is also  used to present a simplified version of  the arbitrary graph, in the hope to reduce the problem to the point where efficient algorithm can be applied. A well-known network model widely used for modelling wireless ad-hoc networks, a unit disk graph preserves the constraint that two nodes are connected if and only if they are within $r$ distance away. Variations can also be made to suit one's need, such as uniform or non-uniform radius $r$ across the network, but the key concept is that $r$ defines a clear-cut boundary of the connectivity. Although network localization can be modelled this way, and may even be further simplified to aligning nodes in the integer domain, the problem remains NP-Hard \cite{Aspnes2}.

On the other hand, however, NP-Hard is considered from a worst-case perspective. For example, in a non-trilateration graph, if a large number of non-anchors are assisted by only 3 anchors, then NP-Hard stands strong for its complexity. In a practical applications, there may be other conditions and parameters (such as denser networks, larger communication ranges and more anchors), which, although do not reduce the theoretical complexity from NP-Hard to something easier, do help improve the average complexity as we will later show through our proposed search algorithm.

In this paper, based on the unit-disk-graph model, we re-evaluate the relation between graph rigidity and unique positioning, and study practical search algorithms for network localization.    
We point out that the implicit properties of the unit disk graph can provide great benefits in localization. Specifically, while ``being connected'' bears good information on the separation and possible posts of the nodes, ``not being connected'' contains equally important information about their locality in terms of where they cannot be. However, this latter information has been largely ignored. We show that an active exploitation of the forbidden region (as well as the the feasible region) can greatly change the landscape of the problem and significantly expedite the practical search algorithms. We also show that the conventional rigidity requirement is no longer the necessary condition for the uniqueness of the solution. This means that rigid graphs  become only a subset of the total set of networks where a unique solution can be found. Exercising on a grid graph with integer coordinates, we further propose a tree-search algorithm that can find the solution within a manageable complexity. The proposed search is derived from depth-first search, but makes essential use of the multiple conditions and constraints inferred by the unit disk graph, to actively rule out invalid branches and keep the intermediate results consistent. While the problem is still NP-hard, simulations show that the search complexity is significantly lower than if the ``no-connection'' information is simply ignored.



\section{Problem Formulation and Rigidity Condition}
\label{sec:formulation}

We formulate the localization problem as a unit disk graph  $G=\{V,E\}$ consisting of $N=|V|$ vertices (nodes) residing in a 2-dimensional (2D) plane. 
A pair of nodes $(i,j)$ whose distance falls within a uniform radius distance of $r = \{r|r \le 1\}$ is assigned an edge (link) $E_{(i,j)}$ with distance parameter $d_{(i,j)}$. We assume the graph is connected and the vertices are placed uniformly at random in the area. We also assume that a special set of $M$ know their absolute locations, where $3\le M < N$ and these $M$ nodes do not lie in the same line, which  is the necessary condition to avoid ambiguity from simple transformation like rotating, flipping, or shifting the graph \cite{Graver}. A realization $p(V_i)$ of $G$ is a point assignment of a node $V_i$ into the plane, and is specified by a coordinate vector $x_i$. For simplicity, we further constrain $x_i$ to be aligned at integer coordinates only and use the Euclidean distance square $d^2_{(i,j)}$ as the distance metric. Such a simplification transforms the graph to a grid graph that deals with integer-valued distances only, but does not lose the generality of a random graph (since any real value with a given level of precision can be scaled up to an integer, with the rest of the graph also scaled proportionally). 

Using this model, the network localization problem can in general be stated as follows: Given a graph defined above, find the absolute/unique location of each node: 
\begin{equation}
\begin{array}{cll}
& \mbox{find :  } x_i, &\  \forall i \in 1..N \\
 \mbox{s.t.}   &   g(x_i,x_j) = d^{2}_{(i,j)}, & \  \forall E_{(i,j)} \in E \\
  & x_i \in \mathbb{Z}^2. & \ \\
  & x_i \neq x_j, & \  \forall i, j \\
 \end{array}
 \label{eq:unit1}
\end{equation}  
where $g(x_i,x_j) = ||x_i-x_j||^2$ is the distance metric. This problem formulation can be classified as a non-linear integer optimization problem, and has an NP complexity that is as difficult as any non-linear integer optimization problem  \cite{Aspnes2}. 

Now consider the implication of the unit disk graph. If two nodes do not have connection, it implies that these nodes must be away from each other by more than $r$. This condition can and should be explicitly put in the mathematical formulation, to restrict the search space and minimize the search effort: 
\begin{equation}
\begin{array}{cll}
 & \mbox{find :  } x_i. & \ \forall i \in 1..N \\
 \mbox{s.t.}   &   g(x_i,x_j) = d^{2}_{(i,j)}, & \ \forall E_{(i,j)} \in E \\
  & x_i \in \mathbb{Z}^2, & \\
  & x_i \neq x_j, & \ \forall i, j \\
  & g(x_i,x_j) > r^2, & \ \forall E_{(i,j)} \notin E.\\

 \end{array}
 \label{eq:unit2}
\end{equation}

As we mentioned earlier, there is an equivalence relation between global rigidity and unique realization of a graph. For basic background and formal definition of graph rigidity, we refer readers to \cite{Graver}. A descriptive and informal concept of graph rigidity can be explained as follows: Replace vertices with free-moving joints and edges with rigid string; then rigidity means that the graph does not change when applying a vector of forces to any subset of the joints, aka the distance between two remains the same, for any pair of nodes. 
Now that we have used the unit-disk-graph model, if we strict ourselves to only the conditions in (\ref{eq:unit1}), then global rigidity and unique realization still holds. However, when we follow the constraints in (\ref{eq:unit2}), then a graph with unique realization no longer implies rigidity. It is enough to confirm this finding through a counter-example. Figure \ref{fig:unit1} shows a graph where we can find a unique localization solution without requiring rigidity. Given 3 anchor nodes (solid black) and 5 edges, the non-anchor node (gray) has 2 possible locations under conditions in (\ref{eq:unit1}). However, with the constraints in (\ref{eq:unit2}), only one graph realization (A) is valid, and the other possibility can be ruled out since the non-anchor falls within the transmission range of an anchor (and should otherwise has an edge to it). 


\vspace{-1.3cm}
\begin{figure}[htp]
\centerline{
\includegraphics[width=1.5in]{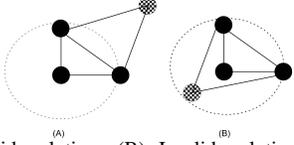}
\vspace{-0.4cm}
}
\caption{(A) Valid solution; (B) Invalid solution, under disk-graph assumption.}
\label{fig:unit1}
\end{figure}
\vspace{-0.4cm}

\section{Tree-Search Algorithm}
\label{sec:search}
\subsection{The Algorithm}

The proposed tree-search localization algorithm is based on the idea of examining the realization possibilities constructively, similarly in flavor to that of list-decoding in channel codes. The anchor nodes form the initial set of realized nodes $R$, and all the non-anchors belong to the un-realized set $U=V\backslash R$. These two sets are collectively viewed as an intermediate instance $T^l_k$ for constructing a complete realization, which is achieved when $U$ is empty and all nodes are in $R$. Our approach is to efficiently go through all possible instances of the realization and  pick the right solution out of that pool. By setting the initial instance $T^0_0$ as a root node of a tree structure, we expand this instance to new instances by moving a non-anchor from the un-realized set to the realized set $R$, at all the possible locations that this non-anchor can possibly be. Each possible case corresponds to a new instance $T^l_k$, where $l$ is the level of the tree and $k=0,1,2\cdots$ is just an index in that level. 
The general process of growing the tree is to 
pick a non-anchor $j\in U$ that has connection to at least one realized node $i\in R$, and priority should be given to the non-anchor that has the most connections with $R$. The distance information $d^2_{(i,j)}$ states that node $j$ must sit within this radius from node $i$, and since location must be aligned in integer coordinates, the possibilities are discrete and countable (and can therefore be indexed by $k=0,1,\cdots$). For efficiency, many of these possibilities should be immediately trimmed out by testing the connectivity/distance conditions between node $j$ and all the other nodes in $R$. Violation of any edge condition or any no-edge condition  should lead to the removal of this instance.  
Clearly, each time a new node is added to $R$ makes the tree to grow one more level. The process continues 
 until all the $(N\!-\!M)$ levels are completed, and the survival instances at the $(N\!-\!M)^{th}$ level are the solutions. We declare all the nodes to be unequivocally localized if there is only one survival. Since all the realization possibilities are either examined or ruled out due to connectivity violations, the tree-search algorithm is optimal. 

It may appear that this tree structure would grow exponentially, with all the possibilities listed at any given level. 
However, this worst case scenario does not really happen, because we make active consistency check using the conditions in (\ref{eq:unit2}). With the conventional network localization formulation in (\ref{eq:unit1}), a sparse graph (i.e. one with relatively small reachability $r$) is usually most time-consuming to localize, due to the very limited  edge information. The proposed strategy as formulated in (\ref{eq:unit2}) is particularly helpful for such cases, because it makes effective use of the lots of no-edge constraints, and therefore drastically reduces the search instances!  


\begin{algorithm}
\caption{Constructing Tree}
\label{alg:1}
\begin{algorithmic}
\STATE $T^0_0 \leftarrow \{R^0_0,U^0_0\}$ 
\FOR{($l=1,l \leq (N\!-\!M),l++$)} 
 \STATE Pick node $n$ from $U^{l-1}_k$, $n$ is a neighbor of $R^{l-1}_k$, $\exists k$  
 \FORALL{$T^{l-1}_k \in Tree$}
 \STATE $X=$ sub locations($n$,$T^{l-1}_k$)
 \FORALL{$x \in X$}
   \STATE add $T^l_x$ as a child of $T^{l-1}_k$  
 \ENDFOR	    
 \ENDFOR
\ENDFOR
 
\end{algorithmic}
\end{algorithm}
\vspace{-0.5cm}

\begin{algorithm}
\caption{Sub Locations($n$,$T$)}
\label{alg:2}
\begin{algorithmic}
\STATE  .\COMMENT {find all valid placements of a node n}
\STATE  . \COMMENT {Return: A point vector $Y=\{y_1,y_2,...\}$}
 \STATE $R \leftarrow$ realization in $T$
 \STATE $m \leftarrow $ neighbor($n$), \ \ $m \in R$
 \STATE $d \leftarrow $ distance $E_{(m,n)}$
 \STATE $X \leftarrow$ all points of radius $d$ at $m$   
 \FORALL{$x_i \in X$ }
 \STATE temporary place $x_i$ in $R$ 
 \IF {validate($x_i$,$R$) using eqn (\ref{eq:unit2}) }
 \STATE add $x_i$ to $Y$
 \ENDIF
 \ENDFOR
 \RETURN $Y$

\end{algorithmic}
\end{algorithm}

\newtheorem{theorem}{Claim}
\begin{theorem}
\label{claim:1}
Algorithm \ref{alg:1} produces at least one valid solution at the leaf node $T^{(N\!-\!M)}_k$, if the graph is realizable.  
\end{theorem}

\begin{proof}
Starting from level 0, the instance $T^{0}_k$ has $M$ nodes (anchors) in the realized set $R$ whose positions are valid. Note that ``being realized'' means that a tentative locality has been assigned to that node, but the assignment may be invalid (incorrectly realized) since violations have not shown. 
Based on the assumption of connected graph, at level 1, the tree constructed from the instance from the previous level must contain at least one instance whose $(M\!+\!1)$ nodes in $R$  are all being {\it correctly} realized. 
Now assume at any level $l$, where $1\!\le \!l\!<\!(N\!-\!M)$, there exists at least one instance $T^{l}_k$ that that has $(M\!+\!l)$ correct locations. Since the tree at level $l\!+\!1$ is produced from all possible placements of a node to the instances from level $l$, then among all the child instances that are inherited from the correct instance at level $l$, there must be at least 1 correct instance whose $(M\!+\!l\!+\!1)$ nodes in $R$ are all correctly realized.   
\end{proof}

\vspace{0.1cm}
\begin{theorem}
\label{claim:2}
Algorithm \ref{alg:1} is optimal, that is, it produces all the valid solutions at the leaf node $T^{(N\!-\!M)}_k$.
\end{theorem}

\begin{proof}
From the algorithm, when the tree expands from level $l$ to $l\!+\!1$, all the possible instances are included (exhaustively). A subtree is trimmed off prematurely if and only if the realized set at that level is already violating some of the geometric distance conditions. Hence, no single valid solution may be left out from this search algorithm.  
\end{proof}

After constructing a complete tree, we can instantly obtain the solution by looking at the instances at the leaf level. If the graph is uniquely realizable, then only one instance at the leaf level. 

\vspace{-0.5cm}
\begin{figure}[htp]
\centerline{
\includegraphics[width=0.5in]{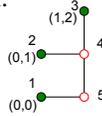}
}
\caption{An example of a connected network with 3 anchors (solid green) and 2 non-anchors(hollow red).}
\label{fig:ex1}
\end{figure}

\vspace{-0.5cm}
\subsection{Example}

 Consider a connected network in Fig. \ref{fig:ex1}. Node 1, 2 and 3 are anchors, while 4 and 5 are non-anchors. Suppose radius $r=1$ and nodes are aligned in integer coordinates. We start the initial realization instance with all the anchors and keep the unknown set aside, resulting in instance $T^0_0$ in Fig. \ref{fig:ex2}(A). Then at level 1, we pick a node, say, node 4, which is connected to the realized sub-tree. Note that node 4 is connected to node 2 and 3 and two possible locations are shown in Fig. \ref{fig:ex2}(B). The invalid instances are omitted to save space. At the next level, we pick node 5 and try to realize it in all the different locations possible. All, except for one, of the tentative realizations will be eliminated due to violations of some rule (distance rule, non-reachability rule, two nodes cannot collide).  The survival is the (only) valid realization. 

\subsection{Implementation, Memory and Complexity}

The above example illustrates how the tree is constructed, and how invalid sub-trees are eliminated as soon as possible. Nevertheless, implementing this algorithm for large networks (lots of non-anchors) can consume a huge amount of memory. To overcome this practicality, instead of width-first search, the {\it depth-first} search can be used along with the dynamic creation of the tree. The maximum space requirement is then limited to $O(N)$. 

The computation complexity is $O(|T|)$, where $|T|$ is the size of the resultant tree, which is directly related to the specific network topology and properties   (such as diameter, average node degree, and communication radius), as well as the order in which the non-anchor nodes are being selected to realize. 
At level $k$, suppose a non-anchor $j$, which connects to some realized node $i$  with distance $d_{(i,j)}$, is chosen to be realized. Then the spanning factor of the tree at this level is upper limited by $D(d_{(i,j)})$, where $D(d_{(i,j)})$ is the number of integer coordinates on the circle  of radius $d_{(i,j)}$. Note that many of these instances may be immediately eliminated after we evaluate all the other rules (formulated in (\ref{eq:unit2})). How much or what portion of these tentative realizations can be immediately discarded is different from case to case and very hard to quantify, but at least we should make $D(d_{(i,j)})$ as small as possible, especially for early levels of $k$. 

\begin{figure}[htp]
\centerline{
\includegraphics[width=0.9in]{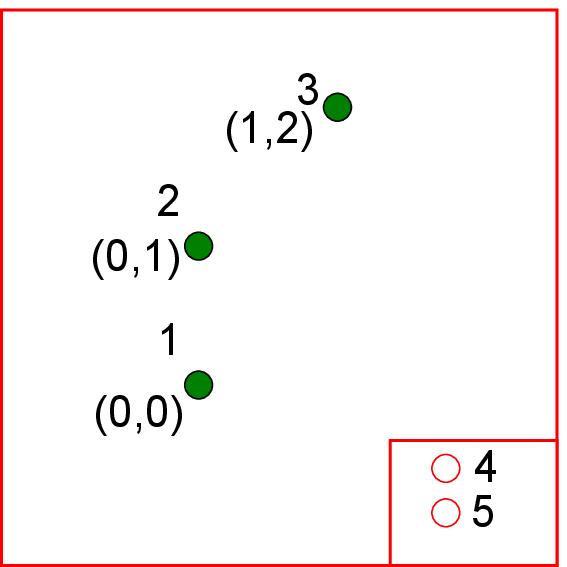} \ 
\includegraphics[width=1.8in]{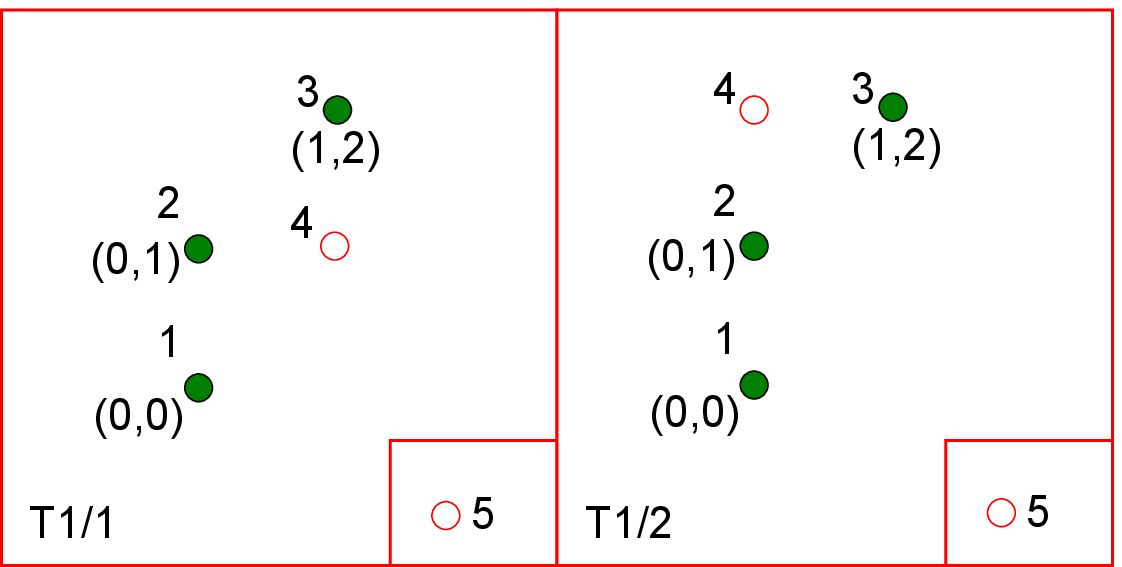}
}
\begin{center}{(A) Level 0. \ \ \ \ \ \ \ \ \ \ \ \ \ \ \ \ \ \ \ \ (B) Level 1. \ \ \ \  \ \ \ }\end{center}
\centerline{
\includegraphics[width=2.7in]{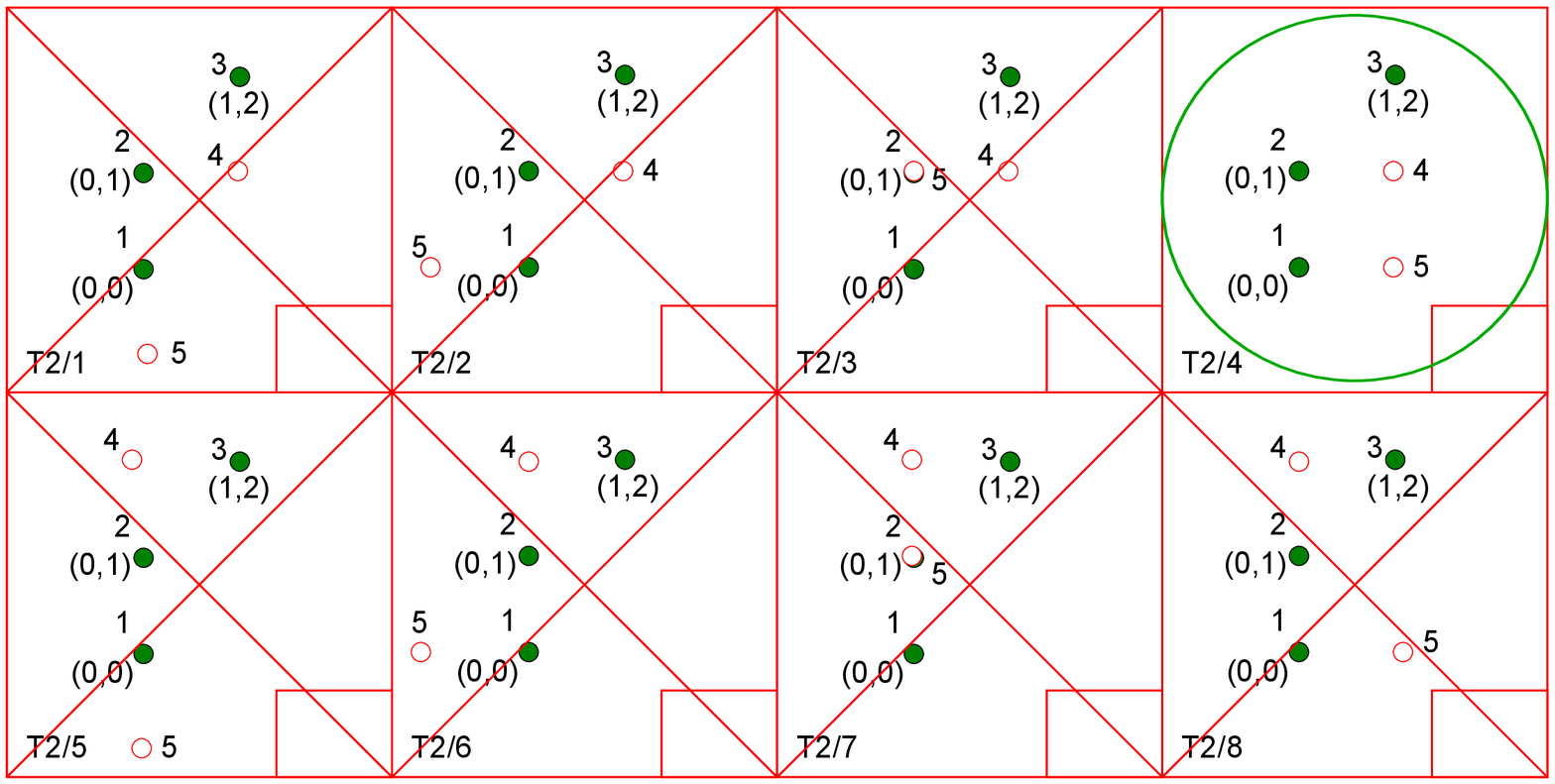}
}
 \begin{center}{(C) Level 2.}
\end{center}
\centerline{
\includegraphics[width=2.7in]{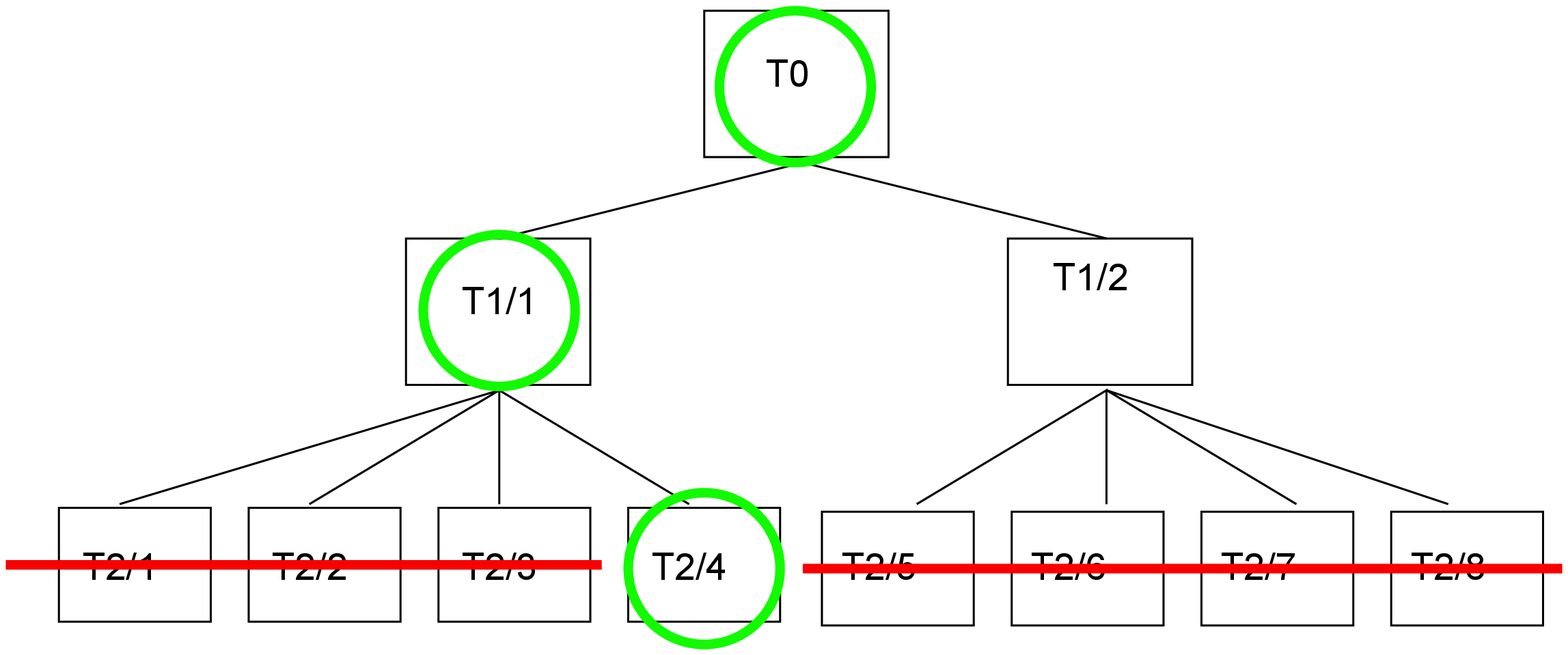}
}\begin{center}{(D) The overall tree $T$.}
\end{center}
\vspace{-.5cm}
\caption{Illustrating example.}
\label{fig:ex2}
\end{figure}
\vspace{-0.3cm}



\section{Simulation Results}
\label{sec:sim}

In our simulation setup,  a set of $N$ nodes are randomly placed in a square grid of $C\times C$. The connectivity is established by radius $r$. A small portion of  nodes ($M$) are randomly picked as anchors. 
Fig. \ref{fig:randnet} shows one example of a generated instance.  
For a given set of parameters $(C, r,N,M)$, we collect the average traversal times (traversal time is equivalent to total number of instances or total number of branches being visited), averaged over many random network realizations using the given parameters. 

In general, a small percentage of anchors makes network localization strenuous. This is particularly so in the conventional case formulated in (\ref{eq:unit1}). By actively exploiting ``useful implicit distance information,'' not only are complexity and  search time drastically reduced, but the grave impact caused by sparse anchors is also drastically relieved.  Fig. \ref{fig:sim1} compares the traversal time by using rules in (\ref{eq:unit1}) (dashed line) and rules in (\ref{eq:unit2}) (solid line). We also perform two different search strategies, one is a random picking of an unknown node when expanding and the other is picking an unknown node that has the most connections to the intermediate graph.  
 The experiments are performed on random graphs with $C\!=\!100$, $r\!=\!25$, $N\!=\!100$, and varying $M$. We see a large magnitude of reduction in traversal time due to (\ref{eq:unit2}), especially when only a very small percentage of anchors present. For example, with $1\%$ of anchors, the traversal time reduces from $252$ to a merely $2.57$ per non-anchor, a speed-up of some 100 times in random picking of a new node  and about 2 times using the most-connected strategy. 
 

\vspace{-0.3cm}
\begin{figure}[htp]
\centerline{
\includegraphics[width=1in]{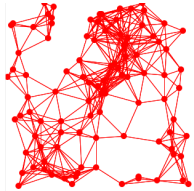}
}\vspace{-0.1cm}
\caption{A random network of 100 nodes in a $100\times 100$ square grid with 
radius $r\!=\!20$}
\label{fig:randnet}
\end{figure}
\vspace{-0.3cm}

In addition to the impact of the percentage of anchors, we also evaluate the impact of node density, or, equivalently, the communication radius $r$. 
Fig. \ref{fig:sim3} shows the case using rules in (\ref{eq:unit1}) and (\ref{eq:unit2}).
 Similar observations can be made. In the conventional case, a sparse network (small communication range and low node degrees) makes localization very difficult. Again, our new problem formulation and search strategy can significantly reduce this negative impact, more than 5 times reduction from 379 to 61 traversal time/node when the radius is 0.2, using the random algorithm while the most-connected algorithm is about 0.5 times better. However, when the network is sufficiently dense (e.g. $r=0.3$), we see that the additional information is no longer the advantage because the larger range implies that a node can reach the anchors easier within a few hops or even in a single hop.  In all cases, the results show that the unit disk graph model with priori information is useful when other information such as the connectivity and anchor nodes are limited.

\vspace{-0.5cm}
\section{Conclusion}
\label{sec:conclusion} 
Unit disk graphs are a popular wireless network model, in which the transmission range $r$ sets a clear cut on node connectivity (and hence distance information). Extracting useful features of this model, we have re-formulated the conventional network localization problem as one that actively exploits the ``no-edge'' information as well as the edge information. We show that the rigidity condition can then be safely relaxed, and still result in unique graph realization. 
We further developed a (depth-first) tree-search algorithm with active branch-trimming,  which runs efficiently in real world setting as the simulation results confirm. We conclude the paper by emphasizing that that information regarding two nodes not being connected can be as important as knowing that they are connected! This implicit information significantly reduces the average complexity.

\begin{figure}[htp]
\centerline{
\includegraphics[width=3.3in,height=2.3in]{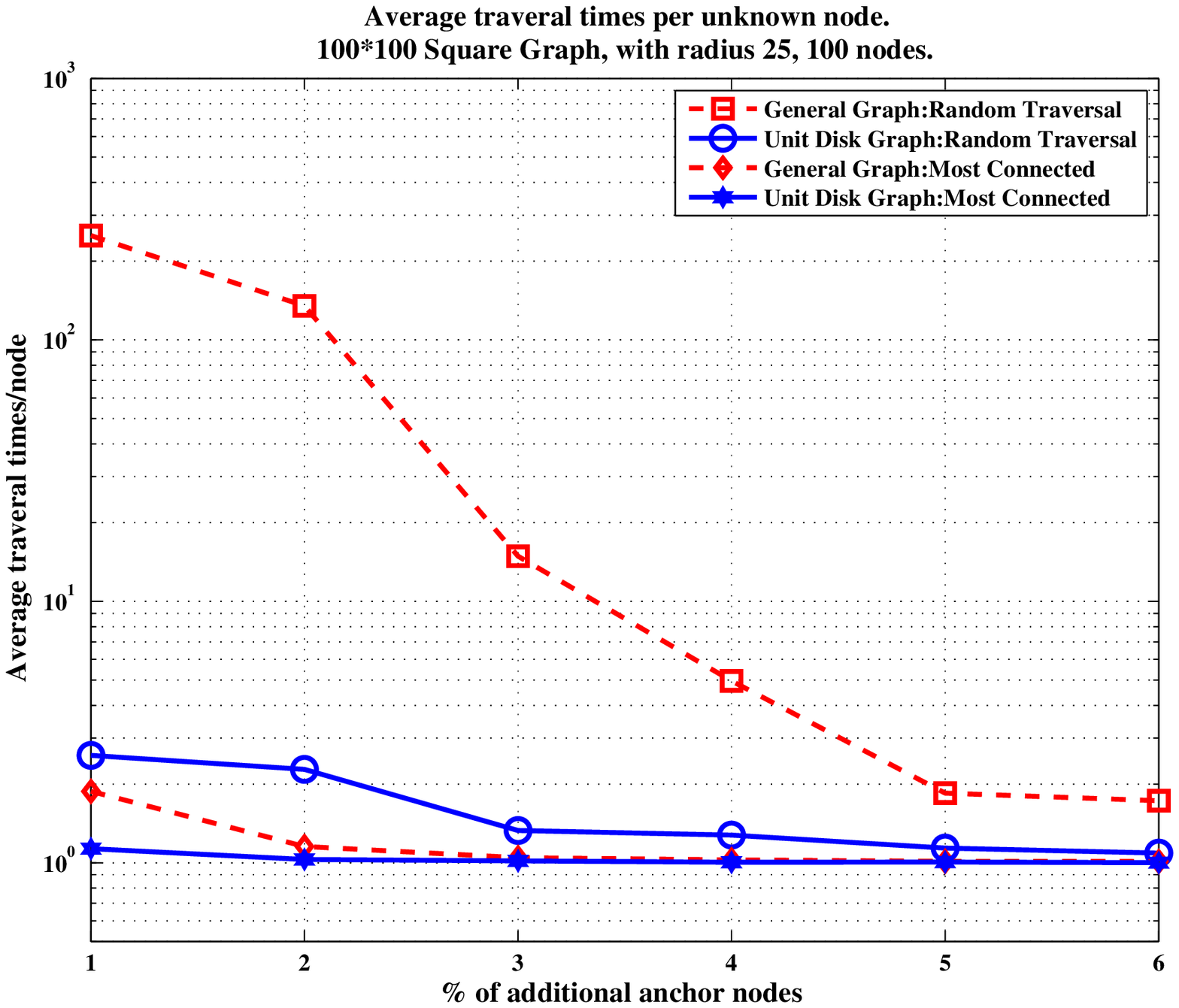}
}
\vspace{-0.5cm}
\caption{Comparison between using (solid line) and not using (dashed line) the no-edge information, for varying percentage of anchors.}
\label{fig:sim1}
\vspace{0.2cm}

\centerline{
\includegraphics[width=3.4in,height=2.3in]{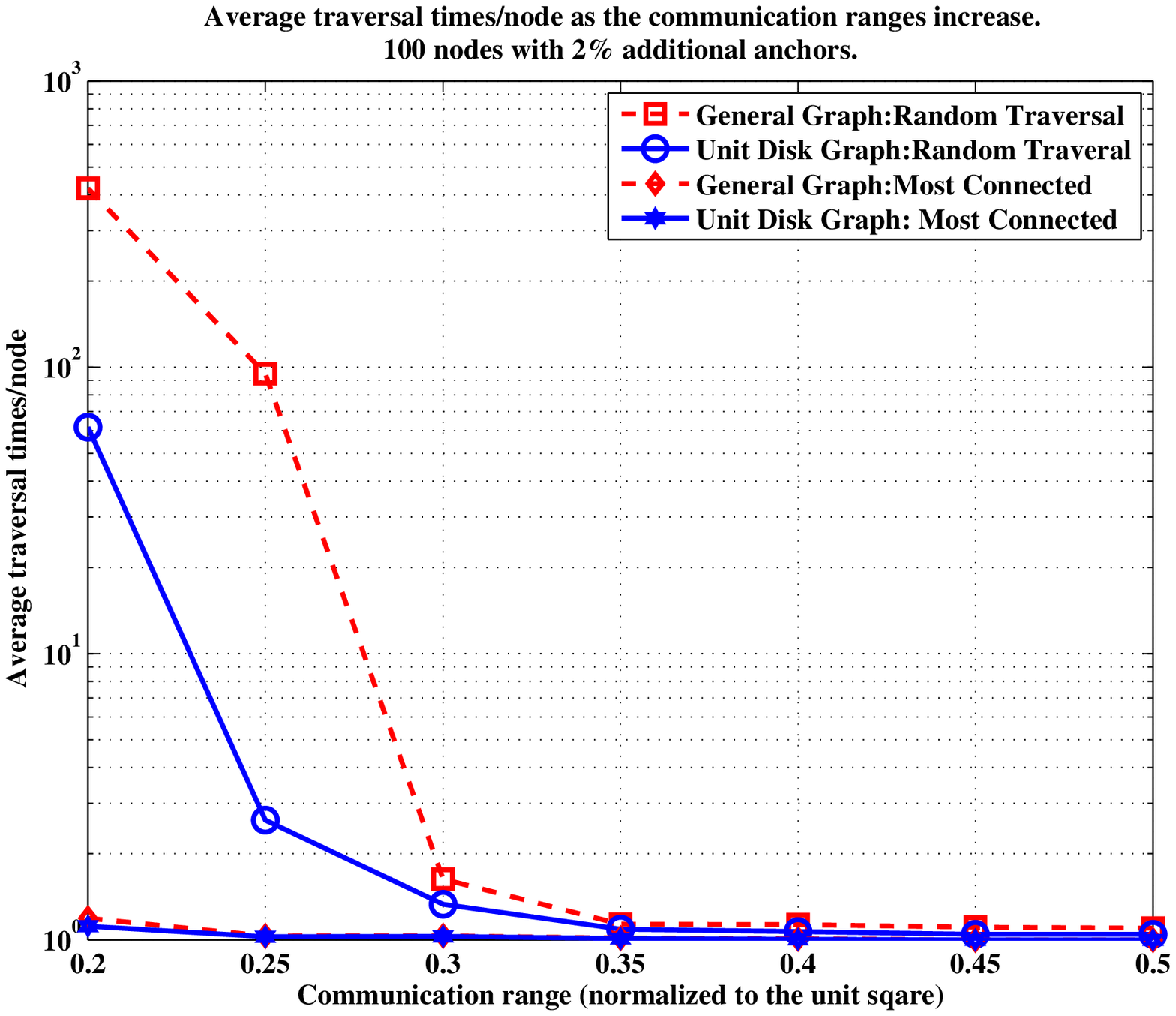} 
}\vspace{-0.5cm}
\caption{Average traversal time as the communication radius $r$ (i.e. density) increases. }
\label{fig:sim3}
\end{figure}

\end{document}